\documentclass[11pt]{article}

\usepackage{booktabs}
\usepackage{amsmath, amssymb, amsthm, thmtools, amsfonts}
\usepackage{bbm}
\usepackage[utf8]{inputenc}
\usepackage{hyperref}

\usepackage[margin=1in]{geometry}

\usepackage{ifthen}
\usepackage{tikz}
\usetikzlibrary{positioning,decorations.pathreplacing}
\usepackage{appendix}
\usepackage{graphicx}
\usepackage{adjustbox}
\usepackage{array}
\usepackage{makecell}
\usepackage{pifont}
\usepackage{amsmath}

\usepackage{color}
\usepackage{pagecolor}
\usepackage[boxed]{algorithm2e}
\usepackage[noend]{algpseudocode}
\usepackage{epstopdf}
\usepackage[textsize=tiny]{todonotes}

\newcommand\eps{\varepsilon}

\usepackage{framed}
\usepackage[framemethod=tikz]{mdframed}
\usepackage[bottom]{footmisc}
\usepackage[shortlabels]{enumitem}
\setitemize{noitemsep,topsep=3pt,parsep=3pt,partopsep=3pt}
\usepackage[font=small]{caption}
\usepackage{xspace}

\newtheorem{theorem}{Theorem}[section]
\newtheorem{lemma}[theorem]{Lemma}
\newtheorem{fact}[theorem]{Fact}
\newtheorem{meta-theorem}[theorem]{Meta-Theorem}

\newtheorem{Remark}[theorem]{Remark}

\newtheorem{proposition}[theorem]{Proposition}

\definecolor{darkgreen}{rgb}{0,0.5,0}
\definecolor{darkred}{rgb}{0.9,0,0}
\definecolor{grayish}{rgb}{0.8,0.8,0.8}
\usepackage{hyperref}
\hypersetup{
    unicode=false,          
    colorlinks=true,        
    linkcolor=darkred,          
    citecolor=darkgreen,        
    filecolor=magenta,      
    urlcolor=cyan           
}
\usepackage[capitalize, nameinlink]{cleveref}
\Crefname{table}{Table}{Tables}
\Crefname{Remark}{Remark}{Remarks}

\algnewcommand\algorithmicswitch{\textbf{switch}}
\algnewcommand\algorithmiccase{\textbf{case}}

\algdef{SE}[SWITCH]{Switch}{EndSwitch}[1]{\algorithmicswitch\ #1\ \algorithmicdo}{\algorithmicend\ \algorithmicswitch}%
\algdef{SE}[CASE]{Case}{EndCase}[1]{\algorithmiccase\ #1}{\algorithmicend\ \algorithmiccase}%
\algtext*{EndSwitch}%
\algtext*{EndCase}%

\newcommand{\bigO}{\mathcal{O}}

\renewcommand{\paragraph}[1]{\vspace{0.15cm}\noindent {\bf #1}.}

\mathchardef\mhyphen="2D

\newcommand{\MPC}{$\mathsf{MPC}$\xspace}

\newcommand{\cclique}{$\mathsf{Congested}$\xspace$\mathsf{Clique}$\xspace}

\newcommand{\poly}{\operatorname{\text{{\rm poly}}}}
\newcommand{\polylog}{\operatorname{\text{{\rm polylog}}}}

\newcommand{\expval}[1]{E\left[#1\right]}

\newcommand{\whp}{with high probability}

\newcommand{\set}[1]{\left\{#1\right\}}
\newcommand{\paren}[1]{\mathopen{}\left(#1\right)\mathclose{}}

\renewcommand{\paragraph}[1]{\vspace{0.15cm}\noindent {\bf #1}.}

\newtheorem{Claim}{Claim}[section]
\begin{document}

\title{\Large Dynamic Graph Algorithms with Batch Updates in the Massively Parallel Computation Model\thanks{Krzysztof Nowicki's research is supported by the Polish National Science Centre project no.\ 2017/25/B/ST6/02010 and by the Foundation for
Polish Science (FNP).}}
\author{Krzysztof Nowicki\thanks{University of Wrocław. Work partially done during an internship at the IBM T.J.\ Watson Research Center.}
\and Krzysztof Onak\thanks{IBM Research.}}

\date{}

\maketitle









\begin{abstract} \small\baselineskip=9pt%
We study dynamic graph algorithms in the Massively Parallel Computation model, which was inspired by practical data processing systems. Our goal is to provide algorithms that can efficiently handle large batches of edge insertions and deletions.

We show algorithms that require fewer rounds to update a solution to problems such as Minimum Spanning Forest, 2-Edge Connected Components, and Maximal Matching than would be required by their static counterparts to compute it from scratch. They work in the most restrictive memory regime, in which local memory per machine is strongly sublinear in the number of graph vertices. Improving on the size of the batch they can handle efficiently would improve on the round complexity of known static algorithms on sparse graphs.

Our algorithms can process batches of updates of size $\Theta(S)$, for Minimum Spanning Forest and 2-Edge Connected Components, and $\Theta(S^{1-\varepsilon})$, for Maximal Matching, in $O(1)$ rounds, where $S$ is the local memory of a single machine.
\end{abstract}

\section{Introduction.}

\subsection{Computing in Parallel: the \MPC model.}

Due to growing amounts of data, processing them in a centralized manner, using a single commodity machine, has often become infeasible. Several approaches addressing this type of challenge have been developed, including multi-machine systems such as MapReduce, Hadoop, and Spark. 
%
%
%
Even though several features of these platforms are different, they all embrace synchronous data processing. This property was captured in the Massively Parallel Computation model, in short \MPC, which was proposed by Karloff et al.~\cite{DBLP:conf/soda/KarloffSV10}.

In a nutshell, in \MPC, $M$ machines perform computation in synchronous rounds. Each machine has local memory $S$, i.e., it can store $S$ words consisting of $\Omega(\log N)$ bits, where $N$ is the number of words of which the input consists. Initially, the input is evenly partitioned between the machines, i.e., each machine receives $\Theta(N / M)$ of them.

Each round of computation consists of the phase of local computation and the phase of communication. During the local computation phase, each machine performs computation on the data stored in its local memory. During the communication phase, each machine can exchange some number of $\bigO(\log N)$ bit messages with any other machine, as long as each machine is a source and destination of at most $\bigO(S)$ messages. The required representation of the output of the algorithm depends on its size. If the output is small, we may require that there is a specific machine that knows it (e.g., for the Graph Connectivity problem, the output can be encoded on a single bit). If the output is large, we allow that the output is distributed across several machines (i.e., for the Minimum Spanning Tree problem, we require that all the edges that belong to the minimum spanning tree are marked).

The most desired goal is to design algorithms that for $S \in \bigO(N^{1-\varepsilon})$ and $S\cdot M \in \bigO(N)$, require a small number of computation rounds. The number of rounds needed to solve a particular problem is called its \emph{round complexity}. The total memory bound is often relaxed. Namely, $S\cdot M \in \omega(N)$ is allowed as long as $S \cdot M$ is not significantly greater than $N$ (for instance, $S \cdot M = N \polylog N$).

\subsection{Evolving Data Sets: Dynamic Algorithms.}
In this paper, we consider \emph{Dynamic Graph Algorithms}~\cite{eppstein1998dynamic}: the data set is an evolving graph and the goal is to maintain a solution to a specific graph problem. Examples of problems considered in the framework include Minimum Spanning Tree \cite{DBLP:journals/jacm/HolmLT01, DBLP:conf/focs/NanongkaiSW17}, Maximal Independent Set \cite{DBLP:conf/soda/AssadiOSS19, DBLP:mm1_focs, DBLP:mm2_focs}, and Maximal / Maximum Matching \cite{DBLP:conf/icalp/ArarCCSW18, DBLP:journals/siamcomp/BaswanaGS18, DBLP:conf/soda/BernsteinFH19, DBLP:conf/icalp/CharikarS18, DBLP:conf/focs/Solomon16, DBLP:journals/talg/NeimanS16}. In this paper we focus on a variant of dynamic graph algorithms in which the change from the $i$-th to $(i+1)$-th graph is described by a batch of $k$ \emph{edge insert} and \emph{edge delete} operations \cite{DBLP:conf/parle/FerraginaL94,DBLP:conf/spaa/AcarABD19}. It may be the case that a solution for the $i$-th data set can be useful for computing a solution for the $(i+1)$-th data set. Hence, our goal is to give an algorithm that processes the update from the $i$-th to $(i+1)$-th data set faster than processing the $(i+1)$-th data set from scratch. 

\subsection{Graph Problems in \MPC and the Motivation for Dynamic Algorithms.} \label{sec:problems}
For graph problems in the \MPC model, rather than expressing the memory of a single machine as a function of the size of the input (i.e., number of edges), it is useful to express it as a function of the number of vertices. For an $n$--vertex $m$--edge graph, we usually consider three variants of the \MPC model, depending on the relation between the local memory of a single machine and the number of vertices in the graph: $S \in \bigO(n^{\alpha})$, for constant $\alpha \in(0,1)$, $S \in \tilde \Theta(n)$, and $S \in \bigO(n^{1+\alpha})$, for constant $\alpha > 0$. Usually, the goal is to provide efficient algorithms with the local memory as small as possible. The reason is that physical infrastructure simulates all machines of the \MPC model on some cluster, and the smaller local memory allows for distributing work between physical machines more evenly. Even though there is no provable separation between these three variants, the best known upper bounds in each of them are different for some graph problems.

In the remaining part of \cref{sec:problems}, we recall known static algorithms for the problems considered in this paper, i.e., the Minimum Spanning Tree problem, 2-Edge Connected Components problem, and Maximal Matching problem. Then we explain the motivation for studying dynamic versions of MPC algorithms for them.

\paragraph{Results related to the Minimum Spanning Forest problem in \MPC}
For the sublinear memory regime, it is known that one can solve the Minimum Spanning Forest problem in $\bigO(\log n)$ rounds, using slightly modified variants of Boruvka's algorithm~\cite{Boruvka}, but there are no known $o(\log n)$ algorithms. Moreover, it is conjectured that if local memory is bounded by $\bigO(n^{\alpha})$ for some constant $\alpha < 1$ and the number of machines is at most polynomial, then deciding whether the input graph is a single cycle or two disjoint cycles---a task that can be solved with a minimum spanning forest algorithm---requires $\Omega(\log n)$ rounds of computation. Throughout the paper, we refer to this problem as the \emph{2-cycle problem} and to the conjecture as the \emph{2-cycle conjecture}.

For the MPC with larger local memory, i.e., $S = \tilde\bigO(n)$ or $S = \Theta(n^{1+\alpha})$, the algorithms were developed independently by two communities: they were phrased either as MPC algorithms \cite{DBLP:conf/spaa/LattanziMSV11,affinity_clustering} or \cclique algorithms  \cite{lotker2005CongestK,Hegeman:2015:TOB:2767386.2767434,GhaffariMSTLogStar,DBLP:conf/soda/Jurdzinski018,Nowicki19} (which also can be deployed in the MPC model \cite{10.1007/978-3-319-09620-9_13}).

Lattanzi et al.~\cite{DBLP:conf/spaa/LattanziMSV11} show how to solve the Minimum Spanning Forest problem in $\bigO(1 / \epsilon)$ rounds with $\bigO(n^{1+\epsilon})$ local memory. Bateni et al.~\cite{affinity_clustering} improve the complexity to $\bigO(\log (1/\eps)+1)$. For $S \in \Theta(n^{1+\alpha})$ both algorithms need only a constant number of rounds to compute a minimum spanning forest. Neither of these papers analyses what the round complexity of their algorithms when deployed on the variant of MPC with $S \in \bigO(n)$. It seems, however, that the algorithm by Lattanzi et al.~\cite{DBLP:conf/spaa/LattanziMSV11} requires $\bigO(\log n)$ rounds, and the algorithm by Bateni et al.~\cite{affinity_clustering} requires $\bigO(\log \log n)$ rounds.

The first algorithm for the Minimum Spanning Forest problem in \cclique is a deterministic $\bigO(\log \log n)$ round algorithm proposed by Lotker et al.~\cite{lotker2005CongestK}. When applied in the MPC model, its round complexity matches that of the algorithm by Bateni et al.~\cite{affinity_clustering}. 

The first $o(\log\log n)$ round algorithm for MPC with $S \in \bigO(n)$ was proposed by Hegeman et al.~\cite{Hegeman:2015:TOB:2767386.2767434}. The authors show that a random sampling approach proposed by Karger, Klein, and Tarjan~\cite{Karger:1995:RLA:201019.201022}, combined with a connectivity algorithm from the streaming model~\cite{ahn2012analyzing} is sufficient to give an algorithm that needs only $\tilde\bigO(n)$ local memory to solve Minimum Spanning Forest in $\bigO(1)$ rounds (or $\bigO(\log\log\log n)$ rounds with $\bigO(n)$ local memory limit). Further improvements (for $S \in \bigO(n)$) were achieved by providing better connectivity algorithms that required fewer rounds with $\bigO(n)$ local memory---firstly, Ghaffari and Parter presented an $\bigO(\log^* n)$ round algorithm~\cite{GhaffariMSTLogStar}, which was extended to an $\bigO(1)$ round algorithm by Jurdziński and Nowicki~\cite{DBLP:conf/soda/Jurdzinski018}.

\paragraph{Results related to  the 2-Edge Connected Components problem}
In the MPC model with linear local memory, one can identify the 2-Edge Connected Components of the input graph by computing a 2-sparse connectivity certificate \cite{IbarakiNagamochi1992}. To compute this certificate, it is enough to solve the Spanning Forest problem twice. Therefore, the 2-Edge Connected Components can be computed in the MPC model with the linear local memory in $\bigO(1)$ rounds \cite{DBLP:conf/soda/Jurdzinski018,Nowicki19}.

For the MPC model with sublinear local memory, one can simulate known PRAM algorithms \cite{DBLP:journals/siamcomp/TarjanV85} in $\bigO(\log n)$ rounds. Furthermore, in \cite{DBLP:conf/icalp/AndoniSZ19}, the authors show that it is possible to find the bi-connected components in $\bigO(\log D \log^2\log_{m/n}n + \log D' \log\log_{m/n}n )$ rounds, where $D$ and $D'$ are the diameter and bi-diameter of the input graph, respectively.

\paragraph{Results related to the Maximal Matching problem}
The first static algorithm for the Maximal Matching problem required $\bigO(n^{1+\epsilon})$ local memory and $\bigO(1 / \varepsilon)$ rounds to finish computation \cite{DBLP:conf/spaa/LattanziMSV11}. For the linear memory regime, some algorithms for $(1+\eps)$-approximation of Maximum Matching were known since 2017~\cite{czumaj2017round,AssadiBBMS19,DBLP:conf/podc/GhaffariGKMR18}, but a $\poly(\log \log n)$-round algorithm for the Maximal Matching was not known until the recent breakthrough of Behnezhad et al.~\cite{maximal_matching_focs}. For the sublinear memory regime, the first sublogarithmic algorithm for the maximal matching was proposed in \cite{DBLP:conf/soda/GhaffariU19,DBLP:journals/corr/abs-1807-08745}.

\paragraph{The motivation for algorithms on dynamic data sets}
\cref{t:compl} summarizes the round complexity of algorithms for the considered problems in MPC with $S \in \bigO(n^{\alpha})$ and $S \in \bigO(n)$.

\begin{table}
\begin{center}
\begin{tabular}{|l|c|c|}\hline
Problem & $\bigO(n^{\alpha})$ local memory & $\bigO(n)$ local memory \\\hline
MSF & \makecell[c]{$\bigO(\log n)$} & \makecell[c]{$\bigO(1)$} \\\hline
2ECC & \makecell[c]{$\bigO(\log n)$} & \makecell[c]{$\bigO(1)$} \\\hline
MM & \makecell[c]{$\tilde\bigO(\sqrt{\log n})$} & \makecell[c]{$\bigO(\log \log n)$} \\\hline

\end{tabular}
\end{center}
\caption{Round complexity of static algorithms\label{t:compl} for Minimum Spanning Forest (MSF), 2-Edge Connected Components (2ECC), and Maximal Matching (MM).}
\end{table}%

As we can see, there is a substantial gap between the round complexities of best known algorithms in those two variants of the MPC model. Furthermore, it may be the case that the known algorithms for MPC with $S \in \bigO(n^{\alpha})$ are optimal---this is in fact widely believed for Minimum Spanning Forest due to the 2-cycle conjecture---and obtaining better bounds is impossible. Hence, if we are limited to static algorithms, we have to choose between the small local memory and better round complexity. 

However, as we show in this paper, considering dynamic variants of these problems allows to obtain algorithms that can efficiently maintain a solution for large datasets in MPC with $S \in \bigO(n^{\alpha})$. Their round complexity is provably lower than that of the corresponding static algorithms. Hence, as long as we can model a dataset we need to process as a dynamic data set undergoing batch updates, we can achieve both small local memory and low round complexity.

\subsection{Dynamic Graph Algorithms in \MPC.}\label{sec:batch_dynamic_mpc}
Italiano et al.~\cite{DBLP:conf/spaa/ItalianoLMP19} initiated the study of dynamic variants of the \MPC model.
They give an algorithm that after one edge update to the input graph can in $\bigO(1)$ rounds of computation maintain solutions to problems such as Connectivity, approximate MST, and Maximal Matching. Their focus is on minimizing the communication and number of machines that are used by the protocol, but they do not address the following issue: if the number of changes to the data set is $\Theta(\log n)$, in order to process them, we again need $\Theta(\log n)$ rounds.

Durfee et al.~\cite{batch_dynamic_soda} address this issue and consider a model in which the updates come in batches, i.e., there is a set of updates that should all be applied at the same time. They propose an algorithm that maintains a solution to Connectivity and can process a batch of size $\Theta(S^{1-\varepsilon})$ in $\bigO(1)$ rounds.

\begin{Remark}\rm
In both papers~\cite{DBLP:conf/spaa/ItalianoLMP19, batch_dynamic_soda}, the authors also try to minimize the overall number of messages that are sent in each round. In this paper, our goal is slightly different. We focus on maximizing the size of the batch that we can process, while minimizing the number of processing rounds. As for the total communication complexity, we allow $\Theta(m)$ global communication during each round, which is usual in the static variant of the \MPC model.
\end{Remark}


\subsection{Our Results.}\label{c_mpc}
In this paper, we further explore the direction suggested by Durfee et al. First, we observe that there is an upper bound on the size of a batch for which we can hope that dynamic algorithms outperform static algorithms. Moreover, we show algorithms maintaining the Minimum Spanning Forest, 2 Edge Connected Components and Maximal Matching of a graph undergoing batch updates. Their round complexity matches the round complexity of static algorithms for MPC with $S \in \bigO(n)$. It seems that any further improvements are either not possible or require progress in developing the algorithms for the static variant of the \MPC model.

\paragraph{The Minimum Spanning Forest problem in the Batch Dynamic \MPC model}
In this problem, we maintain a minimum spanning forest of a dynamically changing graph. Before each batch update, we have a current graph $G$ with a minimum spanning forest $F$. Then, given a batch of updates that change $G$ into a graph $G'$, we have to compute a sequence of updates, such that if we apply them to $F$, then we obtain $F'$ that is a minimum spanning forest of $G'$.

In fact, we show an algorithm that solves a slightly more demanding variant of the problem, defined as follows. Let
\begin{itemize}
\item $G$ be a graph before applying the batch of updates,
\item $U = (u_1, u_2, \dots, u_k)$ be a sequence of updates, i.e., sequence of edge insertions and edge deletions,
\item $F$ be a minimum spanning forest of $G$,
\item $F_x$ be a minimum spanning forest of $G$ after applying the first $x$ updates $(u_1, u_2, \dots, u_x)$.
\end{itemize}
The goal is to return a sequence of updates $U'$ to the minimum spanning forest such that in order to obtain $F_x$, it is enough to execute some prefix of $U'$ on $F$. More precisely, as a result of processing a batch of updates, the algorithm has to return a sequence of updates $U' = (u'_1, u'_2, \dots, u'_{k'})$ and a sequence of indices $(y_1, y_2, \dots, y_k)$ such that $F_x$ is $F$ to which we apply the prefix of returned updates of length $y_x$, i.e., the sequence of updates $(u'_1, u'_2, \dots, u'_{y_x})$.

Our algorithm can process a single batch of updates in $\bigO(\frac{1}{\alpha})$ rounds as long as the size of the batch is $k \in \bigO(n^{\alpha})$. 

In order to give an algorithm for the Batch Dynamic Minimum Spanning Forest problem, we propose a generalization of the \emph{top tree} data structure~\cite{DBLP:journals/talg/AlstrupHLT05}, suited for batch updates in the \MPC model, and a generalization of the standard cycle property of minimum spanning trees~\cite{Tarjan:1983:DSN:3485}. 
Combined with a black box application of an $\bigO(1)$ round algorithm for \MPC model with $\bigO(n)$ memory per machine~\cite{DBLP:conf/soda/Jurdzinski018,Nowicki19}, these tools give an $\bigO(1)$-round algorithm that can process a batch of updates of size $\bigO(S)$.

\paragraph{The 2-Edge Connected Components problem in the Batch Dynamic \MPC model}
In this problem, we have to maintain a 2-Edge Connected Components of a dynamically changing graph undergoing batch updates. There are two ways of formalizing the required output of the dynamic algorithm. In the first, we are required to maintain the set of bridges in the graph, where a \emph{bridge} is an edge that if removed, breaks a connected component into two. More specifically, for a set of updates to be processed, we require that the algorithm outputs a set of updates to the set of bridges. In the second, we  require that each vertex of the graph is assigned a label such that two vertices are assigned the same label if and only if they belong to the same 2-edge connected component.

Our algorithm can solve both variants and is based on the dynamic spanning forest algorithm (which we can maintain via our dynamic MST algorithm), the same top tree data structure that we use for the MST problem, and a rather well known sketching approach to the problem of finding bridges in the graph.

\paragraph{The Maximal Matching problem in the Batch Dynamic \MPC model}
In this problem, we have to maintain a Maximal Matching of a dynamically changing graph. Before each batch update, we have a current graph $G$ with a Maximal Matching $M$. Then, given a new graph $G'$, we have to compute a sequence of updates to $M$ such that applying this sequence results in a matching $M'$ that is maximal in $G'$.

In order to give an algorithm for the Batch Dynamic Maximal Matching problem, we make a few observations. Firstly, we observe that one can reduce the problem of processing a batch of updates to an instance of the Maximal Matching problem on a graph that has a small vertex cover. Then we observe that the algorithm for \MPC with $\bigO(n)$ memory per machine \cite{maximal_matching_focs} can solve this kind of instance, using local memory proportional to the size of this vertex cover. Hence, we can process a batch of updates of size $\bigO(S^{1-\delta})$ in $\bigO(\log 1 / \delta)$ rounds.

\subsection{Comparison with Other Results.}
Two independently developed works due to Anderson et al.~\cite{anderson2020workefficient} and Gilbert and Li \cite{gilbert2020fast} also concern the batch dynamic variant of the MST problem that we address in \cref{sec:dynamic_msf}.

  Even though the paper by Anderson et al.~considers only the \emph{incremental} version of the batch model of MST and aims for work efficient PRAM implementation of the algorithm, the underlying graph property seems to be akin to one used for the algorithm we develop in \cref{sec:dynamic_msf} and the algorithm developed by Gilbert and Li. 

  More precisely, for a set $S$ of $k$ edges to be inserted, all these algorithms identify a set of $\bigO(k)$ edges $R$ that if removed, create a partition into $\bigO(k)$ components $C$ such that computing MST on the graph where the set of vertices are components from $C$ and the edges are edges of $S$ and $R$ finds all replacement edges for the maintained MST. As for the fully dynamic version, i.e., supporting also edge removal operation, we---as well as Gilbert and Li---observe that it creates a set of $\bigO(k)$ relevant connected components, for which we can use an MST algorithm with a linear limit on the memory of a single machine~\cite{DBLP:conf/soda/Jurdzinski018, Nowicki19}. 

\section{Maximum Batch Size.}
As mentioned in \cref{r:batch_size}, our goal is to maximize the batch size and minimize the round complexity of an algorithm, under the assumption that the local memory of a single machine is $S \in \bigO(n^{\alpha})$, for some constant $\alpha < 1$. In this section, we present some evidence suggesting that obtaining the algorithms with better relation between the batch size and the round complexity may be difficult. More precisely, the existence of better batch dynamic algorithms for the Minimum Spanning Forest problem and 2 Edge Connected Components problem goes against the 2 cycle conjecture, and developing better algorithms for batch dynamic Maximal Matching requires developing better static algorithms for a fairly dense graphs. Below, we give a few more details, assuming that the batch size is $k \in \Omega(S^{1+\eps})$, for some constant $\eps > 0$\footnote{The arguments can be also applied for any $k \in \omega(S)$, it is just more convenient to discuss $k \in \Omega(S^{1+\eps})$.}. 

The basic observation is that a sequence of updates can make the input graph empty, and then a single batch of updates can create any $k$ edge graph. Since we assume that $k \in \Omega(S^{1+\eps})$, we can create a $S^{1+\eps/2}$ vertex graph of average degree $\bigO(S^{\eps/2})$. 

To obtain a conditional lower bound for the round complexity of an update for the Minimum Spanning Forest problem and the 2 Edge Connected Components problem, we observe that we can encode any instance of a 2 cycle problem of size $S^{1+\eps/2}$ in a single batch of updates. Hence, we need to solve the 2 cycle problem on a $S^{1+\eps/2}$-vertex graph, using machines with local memory $S$. Since 2 cycle problem can be resolved either by a minimum spanning forest algorithm or 2 edge connectivity algorithm, assuming the 2 cycle conjecture both algorithms require $\Omega(\log S)$ rounds to process a single batch of updates.

For the Maximal Matching problem, the situation is slightly more complex. The only clear implication is that maintaining the maximal matching  under batch updates of size $k \in \Omega(S^{1+\eps})$ is as difficult as computing it for $S^{1+\eps/2}$-vertex graph, with average $S^{\eps/2}$, using machines with local memory $S$. For MPC with $S \in \bigO(n)$ we know that any instance of the maximal matching can be reduced to such fairly dense instance in $\bigO(1)$ rounds \cite{maximal_matching_focs}, however for MPC with $S \in \bigO(n^{\alpha})$, we do not know such a reduction. Hence, even though it seems unlikely, we cannot exclude the possibility of proposing a batch dynamic algorithm that can process a batch of updates of size $k \in \Omega(S^{1+\eps})$ with round complexity better than the round complexity of the static algorithm in MPC with local memory $S \in \bigO(n^{\alpha})$.

\section{Top trees for the \MPC model.}\label{sec:data_structure}
In this section we show a data structure that allows us to maintain a dynamically changing forest in the \MPC model. The basic variant of the data structure allows to process edge insertion and edge deletion updates, in batches of size $k \in \bigO(n^{\alpha})$, in $\bigO(\frac{1}{\alpha})$ rounds. The internal structure of our data structure allows for efficient computation of various queries to the tree, some of them are used in our Batch Dynamic MSF algorithm and 2-Edge Connectivity algorithm.

On the top level, the data structure is a $\Theta(n^{\alpha/2})$-ary variant of \emph{top trees} \cite{DBLP:journals/talg/AlstrupHLT05}, that supports executing updates in batches. As in the paper that introduced top trees, whenever we want to address the input graph we say \emph{underlying tree}, and when we address the tree that is the data structure, we use term \emph{top tree}. 

The top tree data structure is a balanced binary tree, such that each node of this tree corresponds to a cluster, i.e., a connected part of the underlying tree. There is no single fixed implementation of the top tree data structure. Here, rather than adapting an existing implementation to the \MPC model, we propose an implementation that resembles B-trees \cite{journals/csur/Comer79}. More precisely, we guarantee that the depth of the data structure is low by imposing lower and upper bounds on the size (expressed as the number of substructures that are part of a single node) of the internal nodes and by having all basic nodes (i.e. corresponding to the edges of the underlying tree) of the structure on the same depth. Still, since we have that a single node of our data structure corresponds to a cluster (that is to a connected subgraph) of the input tree, we think about this data structure as a variant of top trees.

For our purpose we use $\Theta(n^{\alpha/2})$-ary trees. Such a choice guarantees that we can use various algorithms that traverse the data structure in either bottom-up or top-down fashion in $\bigO(1 / \alpha)$ rounds. Furthermore, it allows us to gather a node of a top tree, its all children and grand children in the memory of a single machine, which is useful when we need to reorganize the internal structure to restore the size invariants of internal nodes. Finally, our \MPC implementation allows to a apply a batch of updates to the underlying tree, in parallel. 

\paragraph{Description}
Our variant of the top tree could be summarized as follows: 
\begin{itemize}
    \item each node of the top tree represents a cluster (a connected set of edges) of the underlying tree;
    \item each edge of a graph is a top tree node of rank $0$, each node of the top tree of rank $r > 0$, that does not represent a whole underlying tree, contains $\Theta(n^{\alpha/2})$ nodes of rank $r-1$;
    \item two nodes of a top tree are sibling nodes if they are children of the same top tree node and share a single vertex of the underlying graph, we call such vertex a boundary vertex;
    \item each node of the top tree remembers its own boundary vertices as well as all boundary vertices of its children;
    \item each node of the top tree has a single vertex called the root of the node;
    \item the root of a top tree node that does not represent whole underlying tree is its border vertex on the path to the root of its top tree parent node.
\end{itemize}
Additionally, each vertex of the underlying graph maintains a set of top-tree nodes that contain this vertex, single top tree node per rank. We call this set a \emph{reference set}. If a top tree node stores a vertex of underlying graph, it stores also its reference set. 

\paragraph{Memory requirements for storing a single top-tree node}
The total size of the description of a top tree node is proportional to the number of its children and the number of all boundary vertices (and their reference sets) stored by this node.

\begin{lemma}\label{lem:boundary_vertices}
The number of boundary vertices of children (which we also call \emph{sub-nodes}) of a single top tree node is $\bigO(n^{\alpha/2})$.
\end{lemma}
\begin{proof}
    Firstly, let us observe that there exists a sub-node with exactly one boundary vertex. If that would not be the case, it would be possible to find a tour, which visits each boundary node only once, and contains cycle.
    
    Assume that there are no sub-nodes with only one boundary vertex. Let us consider a tour that starts in some sub-node and cannot step on a boundary vertex twice. Whenever tour goes into a sub-node with at least two boundary vertices, stepping on a vertex $x$, it can leave this sub-node via some vertex $y \neq x$. This means that such tour either ends in the starting sub-node, or in a sub-node with at least 3 boundary vertices.
    
    The part of the tour that left from this sub-node, and came back without stepping twice on any boundary vertex is a cycle in the underlying graph, which is impossible, because the underlying graph is a tree. 
    
    Knowing that there exist at least one sub-node with one boundary vertex, we can remove from the graph all edges contained in this sub-node, and all non boundary vertices. For the boundary vertex, if it is not shared between two sub-nodes it is not a boundary vertex anymore. Then we can repeat this reasoning, since removing a single sub-node reduces the number of boundary vertices at most by $1$, and at the end there are no boundary vertices, the initial number of boundary vertices has to be no larger than the number of sub-nodes.
\end{proof}

By definition, each node stores only $\bigO(n^{\alpha/2})$ sub-nodes, their boundary vertices and own boundary vertices. By \cref{lem:boundary_vertices} the number of boundary vertices of all children is $\bigO(n^{\alpha/2})$, hence the total number of boundary vertices stored by a single node of the top tree is $\bigO(n^{\alpha/2})$. Since for each of those vertices the size of the reference set is $\bigO(\frac{1}{\alpha})$, the total space requirement for storing the whole top tree node is $\bigO(\frac{1}{\alpha} n^{\alpha/2})$.

\paragraph{Top Trees -- operations}
In the remaining part of this section, we give an implementation of the top tree operations. In the first part of this section, we propose algorithms that execute edge insertion and edge deletion operations. In \cref{subsec:rebalance} we show a protocol that given a top tree that is slightly out of balance transforms it into a top tree that meets the size constraints on all nodes. Then, in \cref{subsec:split_and_link}, we show a reduction from the problem of executing edge insert and edge delete operations to the problem of rebalancing an almost balanced top tree.

In \cref{subsec:queries}, which is the second part of this section, we consider a dynamically changing rooted tree $T$ with labels on each edge (or vertex, which we discuss later). We show that for such tree our data structure allows us to maintain a value of a distributive aggregate function defined on the labels of edges in a subtrees of $T$. 
We say that a function $f$ is distributive aggregative when there exist $f'$ and $g$ and $g'$ such that
\begin{itemize}
\item a value of $f'$ is defined (can be easily computed) for each set consisting of a single edge,
\item $f'(E_1 \cup E_2) = g(f'(E_1), f'(E_2))$, for any two disjoint sets of edges $E_1, E_2$, 
\item $f(S) = g'(f'(S))$.
\end{itemize}
More precisely, we show that given an edge labeled rooted tree $T$, together with its top tree data structure, we can compute the value of a distributive aggregative function in $\bigO(1 / \alpha)$ rounds, for all subtrees of $T$. Furthermore, we show that we can process a batch of path queries. A single path query is a pair of vertices $(u,v)$, and as a result we need to compute the value of a distributive aggregative function on all edges on the path between $u$ and $v$ in $T$.

Here, we formulate the result for the edge labeled rather than for the vertex labeled trees, as the edge is a basic building block of the top tree data structure, and we have that the clusters corresponding to the nodes of the top tree with the same rank are edge disjoint. Furthermore, a small modification of the algorithm for the function defined on edges can extend it to the functions defined on vertices. We show an example of that in \cref{sec:2ec_dmpc}, where we present an algorithm for 2-Edge Connectivity that uses labels defined for vertices.

\subsection{Edit operations -- rebalancing the tree.} \label{subsec:rebalance}
The top tree data structure allows us to modify the underlying forest by deleting and inserting the edges. Generally, in order to provide that the top tree has low depth, we want to preserve the size constraints on the nodes of the top tree. That is, we require that all nodes have size to between $n^{\alpha /2}$ and $c n^{\alpha/2}$. However, during the executions of the algorithms that modify the top tree, we allow the nodes to slightly violate the size constraints. More precisely, we only require that each node is smaller than $c' n^{\alpha/2}$, for some $c' > c$. We call the top tree meeting this relaxed size constraint \emph{almost-balanced} and the nodes larger than $c n^{\alpha/2}$ \emph{overloaded}, and smaller than $n^{\alpha/2}$ \emph{underloaded}.

The common part of edge insertion and edge deletion operations is a protocol that given an almost balanced top tree computes a balanced one.

\begin{lemma}\label{lem:rebalance}
  There is an algorithm that given a tree $G$ and an almost balanced top tree $T$ of $G$ in $\bigO(1 / \alpha)$ rounds of the MPC model computes a balanced top tree $T'$ of $G$.
\end{lemma}
\begin{proof}
In the remaining part of \cref{subsec:rebalance} we present an algorithm that restores the size invariant for the nodes of almost balanced the top tree. 

To restore the size invariants on the nodes, we process the top tree in two phases. In the first phase, we reorganize the top tree to provide that it does not contain nodes that are underloaded, but we possibly introduce some new top tree nodes that are overloaded. Then, in the second phase, we split and rearrange all overloaded nodes, so that in the resulting trees all internal nodes meet the size requirements. 

\paragraph{The first phase of rebalancing}
The first phase is executed in a bottom - top fashion. While processing nodes of rank $r$, we assume that each node contains sub-nodes that are proper instances of a top tree, but possibly with a rank that is smaller than $r-1$. Let $v$ be the node we process, and $r'$ be the highest rank among the sub nodes of $v$. To perform rebalancing, we merge all sub-nodes in such a way, that as a result we get a set of sub-nodes of rank $r'$, each having no underloaded internal nodes. To do so, rather than executing merges right away, we compute a set of \emph{merge} operations that have to be executed, and execute them later, simultaneously, on all levels of the tree. In order to compute the set of merges to be executed, we simulate the following process: while there exists a node $Z$ that has rank $r_1$ that has a sibling node $Y$ of rank $r_2 > r_1$, then virtually merge $Z$ into $Y$, (for the sake of this process replace nodes $Z$ and $Y$ by $Y$). 

At this point, we have a set of sub-nodes of $v$ of rank $r'$, but some of them could be underloaded. Since, as a result, we cannot have that $v$ has underloaded sub-nodes, we have to fix it. To do so, we can look at all sub-nodes and their children (including some sub-nodes of rank $r'-1$ that just got virtually merged into those nodes). The total number of grandchildren of $v$ is at most $\bigO(n^{\alpha})$, hence it fits into memory of a single machine. Then a machine that sees all grandchildren of $v$ and knows which of them are siblings can split them into the groups of size that meets our top tree constraints. If as a result we get only one node, the rank of $v$ becomes $r'$, otherwise it becomes $r'+1$.

\paragraph{The merge operation}
The merge operation is defined for a top tree $T$ of rank $r$ and a set of top trees $T_1, T_2, \dots, T_k$ of ranks $r_1, \dots, r_k$, and set of edges $e_1, e_2, \dots, e_k$  such that for all $i$ $r_i < r$, and $e_i$  is in $T_i$ and has one endpoint in $T$. Furthermore, we assume that $k \in \bigO(n^{\alpha/2})$, as the rebalancing algorithm never needs to merge more than  $\bigO(n^{\alpha/2})$ other trees into tree $T$. The bound on the number of merged trees follows from that the top trees that are being merged into $v$ used to be sibling nodes of some ancestor of $v$. Since there are only $\bigO(\frac{1}{\alpha})$ ancestors of $v$, and each of them has up to $\bigO(n^{\alpha/2})$ sibling nodes, the total number of new nodes in $v$ is $\bigO(\frac{1}{\alpha}n^{\alpha/2})$. 

To execute the merge, for each edge we identify a node $v_i$ of top tree $T$ that has rank $r_i$ and contains endpoint of $e_i$ that is in $T$. To do so, we use the reference sets of the endpoints of $e$. Then we merge $T_i$ with $v_i$, and the merge is executed by the parent node of $v_i$. In other words, we look at a level of $T$ on which we have subnodes of rank $r_i$ and add $T_i$ to the list of children of a properly chosen node $v$ of rank $r_i+1$.

Let us consider a top tree node $v$. The merge operation may require us to add the trees $T_1, \dots, T_{k'}$ to children of $v$. Since any node of a top tree has $\bigO(n^{\alpha/2})$ direct sub-nodes and $k' \leq k \in \bigO(n^{\alpha/2})$, we can gather into memory of a single machine all grandchildren nodes of $v$ and all children of root nodes of $T_1, T_2, \dots, T_k'$. All gathered top tree nodes need to become the grandchildren of $v$. Since a single machine sees all nodes that need to become grandchildren of $v$, and it knows which nodes are siblings, it can compute a partition of all grandchildren of $v$ into children nodes of $v$ in such a way, that all children of $v$ meet the size constraints. This follows from that $T$ cannot have underloaded internal nodes, which means that we have at least $n^{\alpha}$ grandchildren of $v$, and they can be splitted into at least $n^{\alpha/2}$ vertices that are not underloaded. Furthermore, since each $T_i$ introduces at most one new child of $v$, as a result we get that the umber of children of $v$ increased by at most $\bigO(n^{\alpha/2})$. 

Therefore, at the end of rebalancing on all levels, the number of children of any top tree nodes could increase by $\bigO(n^{\alpha/2})$. Therefore, the size constraint on any internal top tree node is violated at most by a constant factor.

\paragraph{Clean-up phase}
As mentioned before, after the first phase of rebalancing, we remove all top tree nodes that are underloaded, but possibly introduce some top tree nodes that could be overloaded (too large by a constant factor $x$). To clean this up, we use following pipelined splitting protocol. 

On all levels of the tree in parallel we do the following: simultaneously, each node that is overloaded splits itself into up to $x$ parts and notifies the parent node about this. Since the splitting happens simultaneously on all levels, it could be the case that the parent node that should receive the notification also gets splitted. Let us consider a parent node $X$ that splits itself and a child node $Y$ that also splits itself. The node $X$ knows the new partition of $X$ and knows which of the new parts contains $Y$. Let us call this part $X'$. Then the node $X$, while creating $X'$, instead including $Y$ as a child of $X'$, it includes all nodes that are created out of $Y$ as a children of $X'$. 

Before we take into account splitting of the child nodes, each newly created node has size within the required constraints for the size of a top tree node. Since each child node splits itself into at most $x$ parts, the size of each newly created node is at most $x$ times too large. Furthermore, in order to create overloaded node, the level lower also had to contain an overloaded node. Therefore, in each step the minimum rank of the overloaded node increases by $1$. Hence, the whole cleanup phase ends in time proportional to the depth of the tree, which is $\bigO\paren{\frac{1}{\alpha}}$.

\paragraph{Summary}
Since the first phase got rid of all underloaded top tree nodes, and the second phase splitted all overloaded top tree nodes (while not introducing new underloaded top tree nodes), after two phases all internal nodes meet the size constraints on the top tree nodes. Furthermore, since all edges of the input graph still are top tree nodes of rank $0$, the depth of the rebalanced tree is $\bigO(1/ \alpha)$.\end{proof}

\subsection{Edit operations -- split and link.}\label{subsec:split_and_link}
In this subsection we propose the algorithms for linking and splitting the top trees. We show how to execute those operation in such a way that the resulting top trees can be balanced by the algorithm from \cref{lem:rebalance} described in \cref{subsec:rebalance}.

\paragraph{Edit operations -- split (delete edge)}
To execute the splitting, for each edge $e$ we identify the part of the tree that gets disconnected from the root. To do so, for each edge we identify \emph{the splitting node} of the top tree. A splitting node of the top tree for an edge $e$ is the highest rank node $v$ of the top tree that contains $e$ and the root-endpoint of $e$ is a boundary vertex of $v$. Then the parent node of $v$ splits the set of sub-nodes into parts that are reachable from the root without going through $v$, and those that are not. The parts that are not reachable from the root, become new instances of the top tree, for the connected components created by the edge delete operation. The original tree may contain now some underloaded internal nodes.  

All those operations can be executed simultaneously, for all edges. As a result, we get a set of data structures, one per connected component, but those structures do not necessarily meet the size constraints of the internal nodes, as some of the internal nodes may be underloaded. At this point, we have a top tree that is almost balanced, hence we can use the rebalancing algorithm from \cref{lem:rebalance}.

\paragraph{Edit operations -- link (add edge)}
On the top level, the tree link operation can be executed as follows: firstly, we treat each edge as a node of rank $0$ and merge it into a node that contains one of its endpoint. Then we merge all trees that share a vertex. Both operations can be understood as a variant of a merge that was executed during rebalancing after the delete operation. However, the number of data structures that are being merged into one vertex could be $\Theta(n^{\alpha})$, which prevents us from direct application of merge protocol from the rebalancing algorithm. 

Therefore, as the first step, a vertex $v$ of rank $r$ that performs the merging treats the nodes to be merged as each of the nodes does not violate the size constraints. Therefore, we do not need to look into children of all merged nodes. Instead, we partition $v$ into $\bigO(n^{\alpha/2})$ nodes of the same rank $r$. Then such $v$ notifies the parent node about the fact that in place of $v$ we now have $\bigO(n^{\alpha/2})$ nodes of rank $r$. After this operation, we have some node that possibly are underloaded (trees that were merged in), and some that are possibly overloaded (a parent node of $v$). 

The good news is that the parent of $v$ can be too large only by a constant factor. This is because in order to increase the parent node size, a child node has to split. In order to split a child into $x > 2$ parts, we have to add at least $\Theta(x \cdot n^{\alpha/2})$ nodes to this child. Therefore, each parent node increases size by at most $\bigO(n^{\alpha / 2})$ from children splitted into two parts, and $\bigO(n^{\alpha}) / \Theta(n^{\alpha/2}) = \bigO(n^{\alpha/2})$ nodes from children splitted into at least three parts. 

Therefore, all internal nodes are either underloaded, or too large by a constant factor. Therefore, we have a top tree that is almost balanced, hence we can use the rebalancing algorithm from \cref{lem:rebalance}.

\subsection{Query operations.}  \label{subsec:queries}
Here, we propose two algorithms that allow us to use the top tree data structure to compute the value of a distributive aggregative function for paths and subtrees of dynamically changing tree. 


\subsubsection{Maintaining values for subtrees.}
\begin{lemma}
Let us consider a rooted tree $T$ for which we have a top tree data structure. For a given distributive aggregative function $f$ we can compute the value of $f$ for all subtrees of $T$ in $\bigO(1 / \alpha)$ rounds.
\end{lemma}
\begin{proof}
  The algorithm computing the value of $f$ is a generalization of the parallel prefix algorithm. The algorithm works in two phases, in the first we traverse the top tree in a bottom-up fashion, computing the value of $f'$ for all clusters of $T$ that correspond to the nodes of the top tree. Then, in the second phase, we traverse the top tree data structure top--down, and compute the value of $f$ for the subtrees rooted in particular vertices.
  
  The first, bottom--up phase has a rather straightforward implementation -- assuming that we know $f'$ for all clusters that are subnodes of a given top tree node $X$, we can compute $f'$ for the cluster represented by $X$ using function $g$. Since we assumed that computing the value of $f'$ for singleton sets can be done, processing the top tree bottom--up computes the $f'$ for all clusters that are represented by the nodes of the top tree.
  
  Here we focus on the second, top--down phase. Starting in the root node $X$ of the top tree, we have the value of $f'$ for all clusters represented by the sub-nodes of $X$. The idea is to compute $f'$ recursively, by processing all subnodes in parallel. 
  
  We assume that when we start to process a top tree node $X$, for each vertex $u$ in the cluster represented by $X$ we know the value of $f'$ on all edges that are in the subtree of $u$ but not in the cluster represented by $X$.
  
  On a single level of recursion, for each vertex $v$ that is a boundary vertex there is at most one subnode $Y$ of $X$ that contains $v$ and the vertices that are closer to the root than $v$. The information we need in order to be able to properly process the recursive call in $Y$ is the value of $f'$ on all edges in the subtree of $v$ that are not in $Y$. Fortunately, for all clusters except $Y$, either all edges from a cluster are in the subtree of $v$ or none of them is. Therefore, computing the value of $f'$ on all edges in the subtree of $v$ that are not in $Y$ can be computed from 
  \begin{itemize}
  \item the values of $f'$ for subnodes,
  \item from the values of $f'$ for vertices $u$ that have part of their subtrees outside of $X$.
  \end{itemize}
  This can be done, because all edges that are in a subtree of $v$ that are not in $Y$ are either in some cluster corresponding to a subnode of $X$ or completely outside of the cluster corresponding to $X$. For all edges that are in the subtree of $v$ and in $X$ (but not in $Y$) we use the values of $f'$ computed for clusters, and for all edges that are outside of $X$ we use the values of $f'$ computed for the boundary vertices of $X$.
\end{proof}

\subsubsection{Batch $(u,v)$-path query operation.}
Here, we propose a protocol that recovers a value of a distributive aggregative function $f$ on a given path. The path is defined by the two endpoints, and the queries can be executed in batches. A single query is executed in a recursive manner. We use that a path from $u$ to $v$ can be splitted into parts that are fully contained by the sub-node of the node answering the query, therefore computing the value of $f$ on the parts of the path connecting $u$ and $v$ contained in sub-nodes is sufficient to compute the heaviest edge on the whole path. In this paragraph we discuss this idea in more details and provide that this approach gives an $\bigO(\frac{1}{\alpha})$ round protocol, that uses only $\bigO(n)$ global memory, while respecting the $\bigO(n^{\alpha})$ bound on a memory of a single machine.

\paragraph{Execution of  batch $(u,v)$-path query operation}
The execution of the protocol in a single node of the top tree is following. For nodes of rank $2$ we can gather whole subtree represented by this node and answer the queries. For nodes of larger rank, for each query pair, we compute a sub-node path between the given vertices, i.e. set of sub-nodes that covers the whole path between the query vertices. If $u$ and $v$ are in the same sub-node, the node recursively asks it for the answer (we call it \emph{direct sub-query}). Otherwise, for each internal subnode on the $(u,v)$ path, it generates a subquery $(u',v')$, where $u'$ and $v'$ are boundary vertices of this structure on the $u,v$ path (we call it \emph{internal sub-query}), while for the structures that hold the endpoints of the  the node asks for $(u,v')$ and $(u',v)$, where u' is a boundary node from node holding $v$ towards node holding $u$, and $v'$ is a boundary node from node holding $u$ towards node holding $v'$ (we call it \emph{endpoint sub-query}). The result is maximum of the answers gathered from recursive sub-queries. 

\paragraph{Communication complexity -- local bounds}
Executing naively all subqueries in parallel could result in large number of subqueries, i.e. we could generate $\Theta(n^{\alpha/2})$ subqueries for a single query, hence $\Theta(n^{3\alpha/2})$ in total. To bypass this issue, we use the fact that the number of distinct subqueries we need to ask is $\bigO(n^{\alpha/2})$, and each sub-node does not receive more subqueries than the parent node. Obviously, the number of \emph{direct subqueries} cannot be larger than the number of original queries, hence it is $\bigO(n^{\alpha})$. Let us consider queries that have endpoints in different sub-nodes: firstly, each query generates two \emph{endpoint sub-queries}, so those contribute $\bigO(n^{\alpha})$ to the total number of sub-queries. Finally, each sub-query is defined by a pair of boundary nodes of a sub-node. Since by \cref{lem:boundary_vertices}, there only $\bigO(n^{\alpha/2})$ boundary vertices in sub nodes, the total number of possible subqueries is at most $\bigO(n^{\alpha/2})^2 = \bigO(n^{\alpha})$. This means that all those subqueries can be sent to sub-nodes in $\bigO(1)$ rounds. Furthermore, for each query, we generate at most one sub-query per sub-node, hence the total number of subqueries is at most constant factor larger than the number of original queries, and each sub-node gets no more than the number of original queries.
 
\paragraph{Communication complexity -- global bound}
Additionally, we can give a bound on the total number of queries that are generated during the execution. Since the query is answered in the top tree nodes of rank $2$ and there are only $\bigO(n^{1-\alpha})$ such nodes, even if all those nodes get a sub-query, for each of all $\bigO(n^{\alpha})$ queries, the total number of sub-queries is $\bigO(n^{1-\alpha})\cdot \bigO(n^{\alpha}) = \bigO(n)$.

\section{Batch Dynamic MST.} \label{sec:dynamic_msf}
In this section we propose an algorithm for the Batch Dynamic Minimum Spanning Forest problem, in the \MPC model. 
\begin{theorem}
\label{thm:msf_dmpc}
Given an $n$ node, $m$ vertex graph, it is possible to solve the Batch Dynamic Minimum Spanning Forest problem in the \MPC model with $\bigO(n^{\alpha})$ memory limit on a single machine using 
\begin{itemize}
\item $\bigO(\frac{1}{\alpha^2}\log n)$ rounds for preprocessing,
\item $\bigO(\frac{1}{\alpha})$ rounds to process a batch of updates of size $\bigO(n^{\alpha})$.
\end{itemize}
\end{theorem}

\begin{proof}
\paragraph{Building the initial top tree data structure}
To compute minimum spanning forest in the \MPC model with $\bigO(n^{\alpha})$ local memory limit, we use a variant of Boruvka's algorithm \cite{Boruvka}. The algorithm maintains a forest that is a subgraph of minimum spanning forest of the input graph, and gradually, in phases, grows this forest. To do so, the algorithm for each vertex chooses a minimum weight outgoing edge, and for each component tosses a coin. To grow the forest, we use all edges that were selected by components that tossed a head and have the other endpoint in a component that tossed a tail. 

This variant of the Boruvka's algorithm with high probability requires only $\bigO(\log n)$ rounds to finish computation, and in a single phase we merge only star shaped graphs, which can be done efficiently in the \MPC model. This algorithm allows us to build the top tree data structure along the way. Whenever we execute a merge in the Boruvka's algorithm, we also merge top trees for all merged underlying trees. 

For a single star, we can sort all the subtrees that have to be merged, in $\bigO(1/\alpha)$ rounds \cite{DBLP:conf/isaac/GoodrichSZ11}. Each machine holding a block of subtrees to be merged initiates a sequence of merges of corresponding top trees. As a result we get some set of partially merged trees. For each of those partially merged trees, we have a new top tree, and the number of trees to be merged in this star is $\Theta(n^{\alpha})$ times smaller. Therefore, iterating this idea $\Theta(1 / \alpha)$ times ends the merging phase, computing corresponding top trees along the way. Each phase of this merging algorithm requires to run merge sorting algorithm, that requires $\bigO(1 / \alpha)$ rounds, and merging operation on top trees, which also takes $\bigO(1 / \alpha)$ rounds. In total, a single phase of Boruvka's algorithm takes $\bigO(1 / \alpha^2)$ rounds, and whole algorithm $\bigO(\frac{1}{\alpha^2} \log n )$ rounds. For a constant $\alpha$ the round complexity is $\bigO(\log n)$.

\paragraph{Processing a batch of updates}
Similarly, as in \cite{batch_dynamic_soda} we firstly focus on processing a batch of updates in a offline manner. That is, we compute the set of updates to the MST that need to be applied after we process whole batch of updates.

Firstly, we process all delete operations. If an edge is deleted, a new edge could be added to the minimum spanning forest, and it is either an edge from the original graph, or it is one of the edges that is a part of the update sequence. The goal of this part is to identify all edges that could be a good replacement edge after some delete, that come from the original graph. 

Then we process all insert operations in order to identify all edges that could violate the cycle property after executing some insert operation. The cycle property says that an edge that is the heaviest edge on a cycle cannot be a part of the minimum spanning forest. Therefore, inserting edge $e$ to the graph, could result with $e$ becoming a part of the minimum spanning forest. Furthermore, if adding $e$ creates a cycle, we need to remove the heaviest edge on this cycle. 

Here, we want to apply whole batch of edge insert and delete operations at once, using only the information provided by the top tree for unaltered variant of the underlying tree. This means that for the $j$th insert we do not really know what would be the spanning forest after execution of $(j-1)$th update operation, which makes identification of the edges that could violate the cycle property a non trivial task. Therefore, we want to identify a set of edges that contains all the edges of the original graph, that could violate the cycle property at each step of the update sequence.

\begin{lemma} \label{lem:ext_cycle_property}
  Let us consider a graph $G$, its minimum spanning forest $F = (V,E_F)$, and some sequence of updates $U$ of length $k$. Let 
  \begin{itemize}
  \item $D = {d_1, d_2, \dots, d_{k_1}}$ be a sequence of edges to delete, 
  \item $R = {r_1, r_2, \dots, r_{k_2}}$ be a set of edges that would be included to the minimum spanning forest, if we execute deletes from $D$,
  \item $I = (i_1, i_2, \dots, i_{k_3})$ is a sequence of edges to be inserted to $G$
  \item $V_U$ be a set of vertices incident to the edges from $D, R, I$. 
  \item $S$ be a set of edges of $F$ that for each pair of vertices $u,v \in V_U$ contains the heaviest edge on the path between $u$ and $v$ in $F$. 
  \item $F' = (V, E_F \setminus (S \cup D) )$.  
  \end{itemize}
  Then we have that $|S| \leq 4k-1$ and the minimum spanning forest after each update consists of components of $F'$connected by some of the edges from $U$, $R$, and $S$.
\end{lemma}

\begin{proof}
  To prove this lemma, we show the following:
  \begin{enumerate}
  \item $|R| \leq |D|$
  \item $|S| \leq 2(|U| + |R|)-1$
  \item after any delete operation, a replacement edge (if exists) is in the set of edges from $U$ or $R$
  \item after each insert operation, an edge violating the cycle property (if exists) is in the set of edges from $U$ or in $R \cup S$.
  \end{enumerate}
  
  \paragraph{Bound on size of R and S (points 1 and 2)}
  The first statement simply says that for each deleted edge, we can have at most one replacement edge. The set $S$ is defined as set of edges, such that for each pair of vertices from $(u,v) \in V_U$ it contains the heaviest edge on the path from $u$ to $v$ over the edges of forest $F$. Let us consider the edges of $F$ from heaviest to the lightest. If the considered edge separates a pair of vertices from $V_U$, we include it to $S$ and remove it from $F$. Since each edge included to $S$ splits a set of vertices into two smaller sets, we can do at most $|V_U|-1$ splittings. Therefore, $|S| \leq |V_U|-1$. Since $|V_U| \leq 2(|D| + |R| + |I|)$ and $|R| \leq |D| \leq |D|+|I|$, we have $|S| \leq |V_U|-1 \leq 4(|D| + |I|)-1 = 4k-1$.
  
  \paragraph{All relevant edges are in sets D, R, I or S (point 3, computing set R)}
  Let us consider an update operation that deletes an edge $e$. If the replacement edge $r$ exists, it is the lightest edge that connects two connected components of the spanning forest obtained by removing $e$. If the edge $r$ is either a member of $I$ or $D$, then we are done. 
  
  Let us assume that's not the case. Let $G^{<r}$ be a graph induced by edges of $E \setminus D$ that are lighter than $r$. If in $G^{<r}$ the endpoints of $e$ are in a single connected component, then in $G^{<r}$ there exist a path between the endpoints of $e$ using only edges lighter than $r$, and none of them got deleted. Therefore, $r$ would not the best replacement edge, which implies that the endpoints of $r$ cannot be in the same connected component of $G^{<r}$. If the endpoints are in separate connected components of $G^{<r}$, then $r$ is the next edge that is to be considered by the Kruskal algorithm, hence it is part of a minimum spanning forest of $(V, E \setminus D)$. Therefore, $r$ is a part of set $R$.
  
  \paragraph{All relevant edges are in sets D, R, I or S (point 4, computing set S)}
  Let us consider an update operation that inserts an edge $e$. The inserted edge could create a cycle in the spanning forest, and our goal is to remove the heaviest edge on such cycle. This cycle could go over several edges that were included during the execution of the update sequence. Any cycle in a graph that is obtained by adding $e$ to a graph obtained by executing a prefix of the update sequence can be splitted into parts, each consisting of paths containing only edges from $U \cup R$ and paths containing only the edges of the original minimum spanning forest. For the first kind of path, we know explicitly all the edges, hence we know an edge with maximum weight on this part of the cycle. The paths of second kind connect the vertices of $V_U$. Therefore, by the definition of $S$ we include maximum weight edge on such path in $S$. Therefore, for any cycle that could be created by inserting an edge, after some prefix of the updates is executed, the edge with maximal weight on such cycle is in $U \cup R \cup S$.  
\end{proof}

Therefore, if we can identify all edges of $R$ and  $S$, and all components of $F'$ incident to those edges, we can simulate the sequence of updates in the local memory of a single machine. In following two paragraphs, we briefly discuss implementation of the approach from \cref{lem:ext_cycle_property} in the \MPC model.

\paragraph{Processing delete operations}
In order to find all replacement edges, we exploit the following:
\begin{fact}
    Let us consider a graph $G$ and its minimum spanning forest $F = \set{T_1, T_2, \dots, T_x}$, and the set $E_d$ of $k$ edges to delete. Let $F' = \set{T'_1, T'_2, \dots, T'_y}$ be a spanning forest of $F$ without the edges from the set $E_d$, and $G'$ be a graph $G$ without the edges from the set $E_d$. Then, all edges of the minimum spanning forest of graph $G'$ that are not in $F$ are between $T'_i$ and $T'_j$, such that $T'_i \notin F$ and $T'_j \notin F$.
\end{fact}

In other words, all edges that could be included in the new minimum spanning forest have to be between the connected components of the maintained forest created by deleting the edges. Since we remove only $\bigO(n^\alpha)$ edges, we create only  $\bigO(n^\alpha)$ connected components in $F$ by removing the edges. Furthermore, having the top tree data structure, we can identify all vertices inside of those connected components in $\bigO(\frac{1}{\alpha})$ rounds. Hence, we can identify all the edges that could be part of the new minimum spanning forest in $\bigO(\frac{1}{\alpha})$ rounds. Having those components and edges, we can treat them as an instance of the Minimum Spanning Forest problem, with $\bigO(n^\alpha)$ vertices, and some number of edges $m' \leq m$. Thus, we can solve it by $\bigO(1)$ round algorithm \cite{DBLP:conf/soda/Jurdzinski018, Nowicki19}, with $\bigO(m')$ global memory and $\bigO(n^\alpha)$ memory limit of a single machine.

\paragraph{Processing insert operations}
In order to compute set $S$, for a given top tree node $v$ and a set of vertices $V_U$ to be separated, we can use a recursive approach. We represent the computed edges as a tree, obtained by contracting all edges that are not the part of the resulting set  of separating edges $S$.  Firstly, let us consider the subnodes of $v$ that correspond to clusters that contain the vertices of $V_U$. Let us assume that we have $x$ such clusters, and let $V_i$ be the set of vertices to be separated that falls to the $i$th cluster. In order to identify all edges that could be the part of the result because they separate vertices of $V_i$ from vertices that are in other cluster, we extend $V_i$ by all boundary vertices of the $i$th cluster. Furthermore, to take into account the edges that do not belong to any of the clusters containing the vertices of $V_U$, but possibly could separate the vertices in different clusters, we also recursively compute a set of edges that separates the boundary vertices, for any cluster that corresponds to a subnode of $v$ that does not contain the vertices of $V_U$. 

To compute the tree representing the edges of $S$, we connect all trees computed by the recursive calls, using the boundary vertices of the clusters. In such tree we contract all edges that are not the heaviest edge between a pair of vertices from $V_U$.

This algorithm can be executed with $\bigO(n^{\alpha})$ local memory, because the result for $y$ points to be separated consists of $y-1$ edges. Hence, from each subproblem we obtain a tree that is either of size $\bigO(n^{\alpha/2})$ or it has the size proportional to the number of vertices of $V_U$ that fall into this subproblem. Therefore, a machine that process $v$ is capable of storing all answers of the recursive problems and can contract all the edges that was received from the recursive calls, but are not the part of the result for $V_U$. As a result, we return $S$ represented by a tree consisting of $|V_U|-1$ edges.

\paragraph{Generating a sequence of updates}
By \cref{lem:ext_cycle_property} the minimum spanning forest at each moment consist of connected components of graph $F'$ which is a spanning forest of the original graph without all edges from $S, R \text{ and } U$. In order to identify connected components of $F'$, we can use the top tree data structure: it is enough to split the minimum spanning forest of the original graph using edges of $S$, and edges of $D$ that are part of $F$. Then the connected components that are relevant to our computation are all connected components of $F'$ that are incident to the edges from $U$, $R$ and $S$.

Since $|U| = k$, $|R| \leq k$ and $|S| \leq 4k-1$, the total number of edges in $U$, $R$ and $S$ is $\bigO(k)$. The number of connected components of $F'$ that are relevant is also $O(k)$. Therefore, the graph consisting of those components and edges can be stored in the memory of a single machine, where the algorithm can simulate the update sequence locally, and generate the output sequence.
\end{proof}

\paragraph{Comparison with the previous result on the Batch Dynamic Connectivity in the \MPC model}
Durfee et al.~\cite{batch_dynamic_soda} show an $\bigO(1)$ round algorithm that can process a batch of size $\bigO(n^{\alpha-\delta})$ using $\tilde \bigO(n^{\alpha})$ communication. Even though our algorithm could use even $\Theta(m)$ messages, it solves more difficult problem, for polynomially larger batch sizes, which seems to be as large as possible (without giving any improvements for algorithms with sublinear space). Also, similarly as the algorithm by Durfee et al.~\cite{batch_dynamic_soda}, our algorithm can be used to solve a slightly more complex variant of the problem called \emph{Adaptive Connectivity} problem (or even \emph{Adaptive Minimum Spanning Forest} problem). In the \emph{Adaptive Connectivity} (\emph{Adaptive Minimum Spanning Forest} ) problem, we ask to execute a sequence of pairs of commands, each consisting from a query and a possible update. The execution of an update is conditioned on the answer to the query. For the Connectivity problem, the query asks whether two vertices are in the same connected component, and for the Minimum Spanning Forest problem, we can additionally ask whether an edge belongs to the minimum spanning forest.

In a nutshell, during the execution of batch variant of the problem, we identified the set of edges that can be parts of the final MST in a way, that allows simulation in the local memory of a single machine, no matter what are the results of the queries. For the edge deletions, the argument is quite similar as in \cite{batch_dynamic_soda}. Since we identified a set of edges that makes graph connected if all deletions would happen, we are prepared to handle any subset of the edge-delete operations. For edge insertions, the argument is quite similar. In order to compute all possible replacement edges, we do not use the weights of the inserted edges. In other words, our algorithm does not really care whether the inserted edges become the part of the minimum spanning forest. It just finds a set of edges that contains all edges that could possibly violate the cycle property, no matter which subset of the edges that are inserted to the graph becomes a part of the minimum spanning forest. Hence, it finds a set of edges that contains all possible replacement edges, for any possible sequence of query answers. 

Durfee et al.\ use the \emph{Adaptive Connectivity} problem to simulate computational circuits. Even though we can process in a constant number of rounds a batch of polynomially larger size, it does not give a huge improvement for the circuit simulation \cite[Corollary 1.2]{batch_dynamic_soda}. The reason is the fact that the algorithm of Durfee et al.\ can process a batch of updates of size $\bigO(n^{\alpha})$ in $\bigO(\log n)$ rounds. Hence a sequence of $\bigO(n)$ updates can be executed in $\bigO(n^{1-\alpha}\log n)$ rounds. Plugging in our result shaves off only the $\log n$ factor.

\section{2-Edge Connected Components in Batch Dynamic \MPC.}\label{sec:2ec_dmpc}
In this section, we propose an algorithm for the 2-Edge Connected Components problem in the Batch Dynamic \MPC model. Our goal is to maintain a set of bridges of a graph undergoing changes, as described in \cref{sec:batch_dynamic_mpc}. Additionally, after each update, our algorithm provides that all vertices have labels denoting their 2-edge connected components, i.e., any two vertices from the same 2-edge connected component have the same label, while any two vertices from different 2-edge connected components have different labels.

\begin{theorem}
\label{thm:2ec_dmpc}
Given an $n$--node $m$--vertex graph $G$, it is possible to maintain a set of bridges of $G$ and a labelling of 2-edge connected components in the \MPC model with $\bigO(n^{\alpha})$ memory limit on a single machine using 
\begin{itemize}
\item $\bigO(\frac{1}{\alpha})$ round to process a batch of updates of length $\bigO(n^{\alpha})$,
\item $\bigO(\frac{1}{\alpha^2}\log n)$ rounds for preprocessing.
\end{itemize}
\end{theorem}

The remaining part of \cref{sec:2ec_dmpc} is a proof of \cref{thm:2ec_dmpc}. Our algorithm for the 2-Edge Connected Components is based on a well known observation. Given a rooted spanning tree $T$ of $G$, let $v^{\downarrow}$ be the tree subtree of $T$ rooted in $v$, $S_T(v)$ be a xor of all identifiers of the edges incident to $v$ and $S_T(v^{\downarrow})$ be a xor of $S_T(u)$ of all $u \in v^{\downarrow}$. Then if an edge $e$ from $v$ in the direction of the root is a bridge in $G$, $S_T(v^{\downarrow})$ is the identifier of $e$. The problem is that it could also happen that for some non bridge edges $e$, $S_T(v^{\downarrow})$ is the identifier of $e$ only because the xor of the identifiers of all other edges outgoing from $v^{\downarrow}$ happened to be $0$. 

One of the possible ways to avoid this problem is to assign random identifiers to the edges, see e.g. \cite[Lemma 2.4]{GhaffariMSTLogStar}. Here, we recall a slightly different approach\footnote{The reason is that \cite[Lemma 2.4]{GhaffariMSTLogStar} uses algorithm with large local computation time. While the formulation of the \MPC model does not require minimizing the local computation, we think that designing algorithms with large local computation time is contrary to the original reason for developing the \MPC algorithms.} based on the implementation of graph sketches from \cite{ahn2012analyzing} and sparse recovery techniques \cite{DBLP:journals/dpd/CormodeF14}. 

\subsection{Sketching and sparse recovery techniques.} \label{sec:tools}
The first ingredient to recall is a linear encoding from \cite{ahn2012analyzing} that allows us to encode a set of edges that are outgoing from (or having exactly one endpoint in) a set of vertices. Firstly, let us define the encoding for sets of vertices that are singleton. Let us consider a vertex $v$ with id $i$. Then the encoding of the edges outgoing from $\set{v}$ is a $n^2$ dimensional vector $S(\set{v}) \in \set{-1,0,1}^{n^2}$ such that:
\begin{itemize}
  \item $S(\set{v})_{j,k} = 1$ iff $i=j$ and there is an edge from $v$ to a vertex with id $k$,
  \item $S(\set{v})_{j,k} = -1$ iff $i=k$ and there is an edge from $v$ to a vertex with id $j$,
  \item $S(\set{v})_{j,k} = 0$ otherwise.
\end{itemize}
Then the encoding of the edges that are outgoing from the set of vertices $V$ is a sum of encodings of all vertices in $V$.

This encoding emulates the behaviour of the xor function. For our purpose, the sum of encodings for all vertices in $v^{\downarrow}$ is an encoding that describes a set of edges that have exactly one endpoint in $v^{\downarrow}$. Therefore, if it encodes a single edge $e$, then $e$ is a bridge going from $v$ in the direction of the root.

The second ingredient to recall is a very simple case of a sparse recovery technique described, e.g., in \cite{DBLP:journals/dpd/CormodeF14}. For our purpose we need another linear transformation that is encodable on $\bigO(\log n)$ bits and allows us to decide whether the encoded vector is empty or not. To achieve that, we can use fingerprints. We treat the coordinates of $v$ as a coordinates of a polynomial over $\mathbb{Z}_p$. Then a fingerprint of a vector is a value of the polynomial for a randomly chosen $z \in \mathbb{Z}_p$. 

Whenever we have a vector that contains only zeros, the value of the polynomial is always zero. For a non zero vector, we have at most $n^2 / p$ probability that we randomly chose a $z$ that is a root of the polynomial described by the vector. Hence, for a non empty vector, with probability at least $1 - n^2 / p$ the fingerprint is non empty. For sufficiently large $p$, that is still polynomial in $n$ the fingerprint allows us to identify a non empty vector with high probability.

To combine those two ingredients for our purpose, firstly we select $p$ and a random $z \in \mathbb{Z}_p$ that are known by all machines. Then we use the observation that in order to compute a fingerprint for whole $S(\set{v})$, we can compute a fingerprint of a vector that corresponds to $v$ with only one incident edge, for each edge that is incident to $v$, and then add all those fingerprints. 

\subsection{The algorithm for the $2$--Edge Connected Components problem.}
In order to obtain the Batch Dynamic \MPC algorithm for the $2$--Edge Connected Components problem we use a slight modification of the xor-based idea. For each subtree we maintain a difference of two fingerprints (computed with the same pair $z,p$). One fingerprint is computed using all edges of the input graph $G$ and the  other fingerprint that is computed using only the edges in the dynamically maintained spanning forest. We can see that if an edge going from a vertex $v$ to the root is a bridge in $G$, then the difference of those two fingerprints is always $0$. Otherwise the difference of the fingerprints is non zero, with high probability.

\begin{Remark}\label{r:batch_size}
  It is possible to skip the idea with subtracting the fingerprints defined on the edges of the subtree, and apply the sparse recovery technique from \cite{DBLP:journals/dpd/CormodeF14} to recover the bridge edge. However, to have a high probability of success we would need to use a sparse recovery data structure of size $\bigO(\log^2 n)$ bits. 
  
  While this additional $\Theta(\log n)$ factor would not impact the round complexity or the batch size that can be handled by our algorithm, it would impact the global memory that is needed by our data structure.
\end{Remark}

The algorithm is as follows. During the execution of the algorithm we maintain a spanning tree $T$ of the dynamically changing input graph $G$, using \cref{thm:2ec_dmpc}. Then, in order to identify the bridge edges, we use the fingerprint based approach. The fingerprints are linear functions, therefore on their own they are distributive aggregative. The small technical nuance is that we define the fingerprint for a set of vertices rather than a set of edges. Therefore, we cannot simply add the values of fingerprints for the set of vertices from different clusters as their sets of vertices may overlap. To fix this issue, we use that whenever we add two fingerprints, the sets of vertices overlap on at most one vertex. Therefore, if we have to compute a fingerprint of two sets $V_1, V_2$, such that $V_1 \cap V_2 = \set{v}$, we can add the fingerprints of $V_1, V_2$, and subtract the fingerprint of $\set{v}$.\footnote{A similar approach is used in the static Minimum Cut approximation algorithm for \MPC \cite{DBLP:conf/podc/GhaffariN20}, however the sketches used in \cite{DBLP:conf/podc/GhaffariN20} need to approximate the size of the cut corresponding to a connected subgraph of a tree, here we use simpler sparse recovery data structure to decide whether a given edge is a bridge.}

As we claimed, besides the bridge identification, our algorithm can also assign to each vertex a label that can be used for verification whether two vertices are in the same 2-edge connected component. Since the possible number of bridges may be $\omega(n^{\alpha})$ we cannot simply remove all bridge edges and recompute the connected components. Therefore, we propose another algorithm that uses the top tree data structure. 

Initially, we mark each bridge edge. Then our goal is to assign the same label to all vertices that are in the same 2-edge connected components (so that the vertices in different 2-edge connected components have different labels). The top level description of the algorithm is as follows. Initially, we assign a label \emph{root} to all vertices. Then, for each marked edge $e$, we identify all vertices from the subtree of $e$ reachable from $e$ over non marked edges and assign $e$ as their label. 

In order to identify all vertices that are in a subtree of $e$ and are reachable over the non marked edges, we again use the cluster decomposition provided by the top tree. Firstly we compute, for each boundary vertex of a top tree node $X$, whether it has the same label as the root of $X$, or some other label. Also, if a boundary vertex has some other label, we compute this label. To do so, we traverse the top tree in a bottom-up fashion. At the bottom, each machine sees the edges of the graph and can verify for each boundary vertex whether it is connected to the root of this cluster with a path consisting of only non marked edges. Then, while processing some top tree node $X$ that has subnodes, we use the answers from subnodes to compute the answer for $X$. More precisely, for each subnode $Y$ of $X$ we can create a graph consisting of the root of $Y$ and the boundary vertices of $Y$ with edges connecting the root of $Y$ to boundary vertices of $Y$ only if they have the same label. Let $G_X$ be the graph created by adding those graphs for all $Y$ that are subnodes of $X$.

Then whenever a boundary vertex of $X$ is reachable from the root of $X$ in $G_X$, it is also reachable in the dynamic tree via path consisting of non marked edges. Therefore, by computing the connected components of $G_X$ we can identify all boundary vertices that have the same label as the root of $X$. For those boundary vertices that are not reachable from the root of $X$ we can assign a label as follows. In order to have a connected component of $G_X$ that is not reachable from the root of $X$, one of the subnodes $Y$ had to have a boundary vertex $b$ that is not reachable from the root of $Y$. Then we know that $b$ has assigned an alternative label. Therefore, for each connected component of $G_X$ that does not contain $b$ we have one vertex with assigned label. Furthermore, this is the only labeled vertex in this component of $G_X$, as each vertex with assigned label has to be disconnected with the root of its cluster (and ancestor in the underlying tree). Therefore, there are no connections between two already labeled vertices. Thus, we can assign the label of $b$ to all other boundary vertices in the same component of $G_X$.

In order to complete the assignment of labels to all remaining vertices of the graph, we use another top--down algorithm. We start in the root of the top tree and we recursively label the vertices in clusters corresponding to all subnodes. Let $X$ be the node we process and $X_r$ be the subnode that contains the root of $X$. Then $X_r$ is labeled recursively with the same label as $X$. For all other subnodes $Y$, if $Y$ is connected to $X_r$ via boundary node that has the same label as the root of $X$, then $Y$ is labeled with the same label that was passed to $X$. Otherwise, $Y$ is labeled with the label that was computed for the root of $Y$ during the bottom-up traversal of the top tree.

\begin{Remark}
Note that since we do not limit our communication, we can compute a new set of fingerprints for each batch of updates, which gives that for each batch our algorithm computes a correct answer with high probability.\footnote{If that would not be the case, once in a while we could select new pair $z,p$ and gradually create another set of fingerprints. Then after processing $n^{1-\alpha}$ batches, we would have a second set of fingerprints that could replace the old one.}
\end{Remark}

\section{Maximal Matching in Batch Dynamic \MPC.}
In this section, we propose an algorithm for the Maximal Matching problem in the Batch Dynamic \MPC model. Our goal is to maintain a maximal matching of a graph undergoing changes, as described in \cref{sec:batch_dynamic_mpc}.

\begin{theorem}
  Let $G$ be a graph, $M$ be a maximal matching of $G$, and $U$ be a sequence of $k$ updates to $G$. Let  $G_U$ be a graph $G$ on which we applied the updates. Given $G, M$, and $U$, we can process $U$ and compute a batch of $\bigO(k)$ updates $U'$, such that applying $U'$ on $M$ gives a maximal matching of $G_U$, in $\bigO(\log 1 / \delta)$ rounds, for $k \in \Theta(S^{1-\delta})$, which for constant $\delta$ gives an $\bigO(1)$ round algorithm, and for $\delta = \bigO(1 / \log n)$, an $\bigO(\log \log n)$ round algorithm.
\end{theorem}
\begin{proof}
Let $M$ be a maximal matching maintained by the algorithm before the updates are applied to the graph, and $M'$ be a matching $M$, from which we removed all edges according to a sequence of updates we process. Then, in order to extend $M'$ to a maximal matching of the graph after all the updates are applied, it is enough to solve a variant of the Maximal Matching problem, with a promise that the graph has a vertex cover of size $\bigO(k)$.
\begin{Claim} \label{c:small_cover}
The set of edges that can extend $M'$ to a maximal matching of the current graph has a vertex cover of size at most $2k$.
\end{Claim}
\begin{proof}
Applying $k$ updates to the graph can remove up to $k$ edges from maintained maximal matching and can insert up to $k$ new edges. Any edge that could extend $M'$ has to be either incident to an edge that got removed during the updates, or has to be a completely new edge. If $M' \subset M$ could be extended by some edge $e$ that is not incident to any edge from $M$, and is not new, then $e$ was present in the graph before the updates. Therefore, $e$ could be added to $M$, which contradicts maximality of $M$. The set of endpoints of $k$ edges (inserted and deleted) has size at most $2k$, hence all edges that can extend $M'$ are covered by at most $2k$ vertices.
\end{proof}

To solve this special variant of the Maximal Matching problem, we use an $\bigO(\log \log n)$ round Maximal Matching algorithm designed for the static variant of the \MPC model, with a linear memory of a single machine \cite{maximal_matching_focs}. More precisely, we show that the memory requirement of a single machine is actually linear with respect to the size of the vertex cover of the input graph, which combined with \cref{c:small_cover} and some small adjustments, is enough to provide an algorithm for Maximal Matching in Batch Dynamic \MPC. 

The static algorithm builds a maximal matching in phases. It starts with an empty matching, and input graph as a source of edges that can extend this matching. In a single phase it considers only a \emph{residual graph}, i.e. a graph in which there are only the edges that still can extend the matching obtained at the end of the previous phase, and extends the matching by some set of edges from this residual graphs. In order to obtain $\bigO(\log \log n)$ round complexity, the authors show that one phase can be implemented in $\bigO(1)$ rounds, and that if $\Delta_i$ is the maximal degree of the residual graph after the $i$th phase, then $\Delta_i \in \bigO(\Delta_{i-1}^{1-\Omega(1)})$, with exponentially high probability.

Here, we show that a single phase of the static algorithm, assuming existence of a vertex cover of size $k$, can be executed in $\bigO(1)$ rounds, using machines with local memory $\bigO(k)$, and guarantees that degrees of vertices in the cover drop from $\Delta$ to $\Delta^{1 - \Omega(1)}$, with exponentially high probability.

\begin{lemma}\label{lem:single_phase_mm}
  Given a graph $G$ of maximal degree $\Delta$, and its vertex cover of size $k$, it is possible to compute a matching $M$, in $\bigO(1)$ rounds of the \MPC model, with the limit on local memory $S \in \bigO(k)$, such that the maximum degree of a vertex in the cover of the residual graph is $\bigO(\Delta^{0.999})$.
\end{lemma}

\begin{proof}
  The proof is basically repetition of the analysis provided in \cite{maximal_matching_focs}, in which we use that the number of edges in a random subgraph may be bounded by estimating the number of cover vertices in this subgraph and their degrees -- as a result w get that the required memory is linear in $k$ rather than in $n$. 
  
  For the sake of completeness, we briefly recall analysis of the maximal matching algorithm \cite{maximal_matching_focs}, slightly adjusted to the case, where the algorithm works on machines with $\bigO(k)$ memory, and the input graph has a vertex cover of size $k$. Still, in order to get a more detailed version of the analysis, we recommend to look it up in \cite{maximal_matching_focs}.

The algorithm from \cite{maximal_matching_focs} consists of phases, and in each phase the maximal degree drops significantly (form $\Delta$ to $\Delta^{1-\Omega(1)}$), which implies the $\bigO(\log \log \Delta)$ round complexity. The key part of our variant of the analysis is that we give a bound on the memory requirements of the algorithm, by considering only the vertices from the given vertex cover. Furthermore, we show that the bound on the maximal degree in the cover set follows from analysis of \cite{maximal_matching_focs}. We also obtain similar bound on the degree in the whole graph, which is necessary if we want to use unaltered variant of the analysis.  To do so, we use a simple observation: after executing a phase of the algorithm, the number of high degree vertices in the whole graph is $\bigO(k)$. Therefore, running it twice, the first time to reduce the degrees of the vertices in the given cover, and the second time to reduce the degrees in the high-degree residual graph, provides a guarantee that the maximal degree of whole graph is reduced. 

In the remaining part of this section we focus on a single phase of the algorithm by \cite{maximal_matching_focs}, and show that we can implement it in a way that requires only $\bigO(1)$ rounds, and works with $\bigO(k)$ memory limit on a single machine. A single phase of the algorithm consists of two stages.

\paragraph{Stage 1}
The first stage of a phase is based on the following algorithm:
\begin{itemize}
\item partition vertices into $x = \Delta^{0.1}$ random sets $V_1, V_2, \dots V_{x}$, and $V_i^c$ denotes intersection of $V_i$ with the cover set
\item sample the edges with probability $1/\Delta^{0.85}$, let $G_i$ be a graph with the vertex set $V_i$ with edges with both endpoints in $V_i$ that survived the sampling
\item compute greedy maximal matching $M_i$ of $G_i$ , for a random ordering of the edges
\item extend the result matching $M$ by $\bigcup M_i$, remove all edges incident to $\bigcup M_i$ from the graph
\end{itemize}
Our claim is that the set of edges of each $G_i$ is $\bigO(k)$. Furthermore, the number of vertices of degree larger than $\Delta^{0.99}$ in the cover set is $\bigO(k/\Delta^{0.03})$. The bound follows from the analysis provided \cite{maximal_matching_focs} slightly modified to take into account the fact that we want to give a bound dependent on the size of the cover rather than the number of all vertices.

\paragraph{Analysis: memory requirement of Stage 1}
Firstly, we show that the algorithm can be executed using only machines with $\bigO(k)$ local memory.
\begin{Claim}\label{claim:memory_alg_1}\ 
    \begin{itemize}
        \item For every $i \in [x]$, $|V_i^c| \in \Theta(k / \Delta^{0.1})$.
        \item For every $i$, graph $G_i$ contains $\bigO(k / \Delta^{0.05})$ edges.
    \end{itemize}
\end{Claim}

\begin{proof}
The analysis is more or less the same as in \cite{maximal_matching_focs}, although we exploit the existence of a vertex cover of size $k$,\footnote{To get exponentially high probability bounds, we consider $k \in \Omega(n^{\epsilon})$.} to show a bound that depends on $k$ rather than $n$.
    
The first property follows from Chernoff bound, as $E[|V_i^c|] \in \bigO(k / \Delta^{0.1})$. For the second property, we consider two cases: $\Delta  > n^{0.01}$ and $\Delta \leq n^{0.01}$.

For larger values of $\Delta$, it is enough to observe that the expected degree of a vertex is $\bigO(\Delta^{0.05})$, and for $\Delta$ that is polynomial in $n$ it holds w.e.h.p. Since all edges have at least one endpoint in the cover set, we can bound the number of edges in $G_i$ by bounding the number of edges in $G_i$ that are incident to vertices of $V_i^c$. Since we already have that $|V_i^c| \in \Theta(k / \Delta^{0.1})$, w.e.h.p., the number of edges in $G_i$ is at most $\bigO(\Delta^{0.05}) \cdot \Theta(k / \Delta^{0.1}) = \bigO(k / \Delta^{0.05})$.

For smaller values of $\Delta$, we still have the expected number of edges that is $\bigO(k / \Delta^{0.1})$. Since the number of edges is a function of $\bigO(k \Delta)$ independent random variables, and outcome of a single variable changes the outcome of hole function by $\Delta$ (number of edges is a $\Delta$-Lipschitz function), we can use the bounded differences inequality [\cref{prop:bounded_diff_ineq}].

\begin{proposition}\label{prop:bounded_diff_ineq}[Bounded differences inequality, formulated in this useful way in \cite{maximal_matching_focs}]\\
Let $f$ be a $\lambda$-Lipschitz function on $y$ variables, and let $X = (X_1, \dots, X_y)$ be vector of $y$ independent (not necessarily identically distributed) random variables. Then, w.e.h.p. (with respect to a parameter $n$),
$$f(X) \leq \expval{f(X)} + \lambda n^{0.01}\sqrt{y}$$
\end{proposition}

Applying this inequality gives that the number of edges is $\bigO(k / \Delta^{0.05}) + \Delta k^{0.01} \sqrt{k\Delta} = \bigO(k / \Delta^{0.05})$ with probability $1 - \exp(k^{0.01})$. As long as $k$ is polynomial in $n$, it is also exponentially high probability with respect to $n$.
\end{proof}

As a corollary from \cref{claim:memory_alg_1} we have that the number of edges of each $G_i$ is $\bigO(k)$, hence even though the number of vertices of the graph $G_i$ could be $\omega(k)$, we still can gather all relevant vertices and edges in the memory of a single machine and compute a matching of $G_i$.

\paragraph{Analysis: degree reduction of Stage 1}
Now we briefly sketch the part of the proof that shows that we can use this algorithm, to obtain a residual graph, in which the number of vertices in the cover, that have high degree larger than $\Delta^{0.99}$, is $\bigO(k / \Delta^{0.03})$, w.e.h.p. 

Firstly, let us consider single execution , for each vertex, the proof analyses what is the number of edges that go between this vertex and a single group of vertices (in random partition). By symmetry, it is enough to analyze this for a fixed group, e.g. $Z_{1}$. The main idea is to splits vertices of the graph into two kinds: vertices that have small impact on the number of edges in the residual graph (variance of $Z_{v,1}$ $\bigO(\Delta^{1.4})$) and all remaining vertices. 

For the first kind of vertices, the authors show that the probability of having degree larger than $\Delta^{0.99}$ is $\bigO(\Delta^{-0.03})$, and the analysis does not depend on value of $n$ or $k$. This gives that if there are $k'$ such vertices in the cover, the expected number of such vertices in the cover of residual graph with degree larger than $\Delta^{0.99}$ is $\bigO(k'\Delta^{-0.03})$.

For the second kind of vertices, the authors show that the expected sum of variances of all vertices in some set $S$ is $\bigO( m/\Delta \cdot 2\Delta \cdot \Delta^{0.15} + |S|\Delta^{1.15})$ \footnote{in the original analysis $S$ was a set of all vertices, and $m / \Delta$ is replaced by $n$, as such bound was good enough for their purpose}. If we take $m = \Delta k$ and $S$ to be a set of vertices in the cover, we get that the expected sum of variances of all vertices in the cover is $\bigO(k \Delta^{1.15})$. From here, we have that in expectation, there can be at most $\frac{k}{\Delta^{0.25}}$ vertices that do not qualify as the first kind. 

Therefore, the expected number of vertices in the cover of residual graph with degree larger than $\Delta^{0.99}$ is $\bigO(k'\Delta^{-0.03})+ \bigO(k\Delta^{-0.25}) = \bigO(k \Delta^{-0.03})$. 

\paragraph{Reducing degree with exponentially high probability}
If at the beginning, there are at most $k / \Delta^{0.03}$ vertices of degree at least $\Delta^{0.99}$, then we are done. Otherwise, we have that the number of edges in the graph is at least $k \Delta^{0.96}$. There are two cases to consider. 

If $\Delta < k^{0.1}$, then we can use inequality from \cref{prop:bounded_diff_ineq}, which gives exponentially high probability of obtaining a residual graph with at most $k / \Delta^{0.03}$ vertices of degree at least $\Delta^{0.99}$, in the cover.

If $\Delta > k^{0.1}$, then we can run $k^{0.05}$ instances of the algorithm. A single instance, by Markov inequality, has a constant probability, that the number of vertices in the cover that have degree larger than $\Delta^{0.99}$ is $\bigO(k / \Delta^{0.03})$. Repeating this $k^{0.05}$ times, in parallel, gives an exponentially high probability that in one of those instances we have $\bigO(k / \Delta^{0.03})$. vertices of degree larger than $\Delta^{0.99}$ in the cover of residual graph. Since one of the steps of the algorithm is to subsample the edges of the input graph, we can bound the total space required by all instances by $k^{0.05} \cdot \bigO(k\Delta \cdot \Delta^{-0.85}) \leq \Delta^{0.5} \cdot \bigO(k\Delta \cdot \Delta^{-0.85} = \bigO(k\Delta^{0.65})$, which is smaller than $\bigO(k \Delta^{0.96})$ space required for just storing the edges of the input graph. Therefore, all instances of the algorithm can be run in parallel, without any additional space.

\paragraph{Stage 2}
In the second phase of the algorithm, we sample each edge in the residual graph incident to a vertex with degree larger than $\Delta^{0.999}$ with probability $q = 1 / \Delta^{0.99}$. Since after the first stage, we have at most $\bigO(k / \Delta ^{0.03})$ vertices of high degree, the total number of edges is $\bigO(k\Delta^{0.99}) + \bigO(k / \Delta ^{0.03} \cdot \Delta) = \bigO(k \Delta^{0.99})$. Sampling each edge with probability $q$ gives us a set of edges that has size $\bigO(k)$, w.e.h.p. Thus, it can be gathered in the memory of a single machine, where we can compute a greedy maximal matching of this subsampled graph. 

We compute a greedy maximal matching on the set of obtained edges, add it to the final solution and remove all edges incident to computed matching. By properties of greedy maximal matchings of uniformly random subgraphs we have that for reach vertex, its degree in the residual graph is $\bigO(\Delta^{0.991})$ with probability $1 - \exp(\poly(\Delta))$. This implies that the expected number of high degree vertices in the cover of residual graph is $k \cdot \exp(\poly(\Delta))$.

If $\Delta > k^{0.01}$, by Markov inequality the degree of all vertices in the cover of residual graph is $\bigO(\Delta^{0.991})$, with exponentially high probability . For $\Delta \leq k^{0.01}$, we use that adding / removing a single edge changes the match status of $\bigO(\Delta)$ vertices \cite{maximal_matching_focs}, which means that the number of high degree vertices in residual graph is a $\bigO(\Delta)$-Lipschitz function of $k\Delta$ random variables, which allows us to apply inequality from \cref{prop:bounded_diff_ineq}. This gives us that the number of vertices in the cover of residual graph with degree higher than $\bigO(\Delta^{0.991})$ is $k \cdot \exp(\poly(\Delta)) + \sqrt{k\Delta} \Delta k^{0.01} = \bigO(k^{0.51} \Delta^{1.5} + k \cdot \exp(\poly(\Delta)))$. By assumption $\Delta < k^{0.01}$, this is easily $\bigO(k / \Delta^2)$.

This means, that the remaining number of high degree vertices in the cover of residual graph is so small, that we can gather all edges incident to those vertices in the memory of a single machine, and greedily match them, which eliminates all remaining high degree vertices from the cover.
\end{proof}

Some part of the proof of \cref{lem:single_phase_mm} require that degree of all vertices is small, therefore we have to somehow handle the fact that reducing degree of the vertices in the cover may be not enough. To handle this issue, we use the third observation, that the number of remaining high degree vertices is $\bigO(k)$.

\begin{Claim}\label{c:small_number_high_degree}
  After executing a single phase of algorithm from \cite{maximal_matching_focs}, on a graph $G$ with maximum degree $\Delta$, the number of vertices in the residual graph that have degree larger than $\Delta^{0.999}$ is $\bigO(k)$.
\end{Claim}

\begin{proof}
  By \cref{lem:single_phase_mm}, we have that the degree of vertices in the cover of residual graph is $\bigO(\Delta^{0.999})$, \whp. Therefore, the number of edges in the whole graph is at most $\bigO(k) \cdot \bigO(\Delta^{0.999})$, \whp. This means, that in the residual graph there are at most $\bigO(k) \cdot \bigO(\Delta^{0.999}) / \Omega(\Delta^{0.999}) = \bigO(k)$ vertices of degree $\Omega(\Delta^{0.999})$, \whp.
\end{proof}

Having \cref{lem:single_phase_mm} and \cref{c:small_number_high_degree} is enough to give an algorithm that guarantees that in $\bigO(1)$ rounds the degree of whole graph drops significantly \whp. Firstly, we run a phase of the algorithm from \cref{lem:single_phase_mm}. Then we identify all vertices with degrees $\Omega(\Delta^{0.999})$, and add them to the cover. Since there are $\bigO(k)$ such vertices, the size of the cover is still $\bigO(k)$. Now we run the algorithm on the graph with this extended cover -- by construction of this cover all vertices of degree $\Omega(\Delta^{0.999})$ are already in this cover, and running the algorithm from \cref{lem:single_phase_mm} reduces those degrees to $\bigO(\Delta^{0.999})$. As a result we get a residual graph in which all vertices have degree $\bigO(\Delta^{0.999})$

\paragraph{Remark}
We use a single phase of the algorithm from \cite{maximal_matching_focs} in a black box manner (twice, to get reduction of degree in whole graph), showing that its memory requirements can be expressed as a function of the vertex cover. Similarly as in the paper \cite{maximal_matching_focs}, if we are given larger memory, we can use larger probabilities of sampling, which yields faster algorithms, more precisely for $k \in \Theta(S^{1-\varepsilon})$ the round complexity is $\bigO(\log(1/\varepsilon))$.
\end{proof}

\bibliographystyle{abbrv} 
\bibliography{ref}

\end{document}